\documentclass[reqno,12pt]{amsart}

\usepackage{amsmath}
\usepackage{latexsym}
\usepackage{amssymb}
\usepackage{hyperref}
\usepackage{graphicx}
\usepackage{epstopdf}
\usepackage{ifpdf}
\ifpdf
\fi

\makeatletter
 
 \newcommand{\Rmnum}[1]{\expandafter\@slowromancap\romannumeral #1@}
 \makeatother

\newtheorem{theorem}{Theorem}[section]

\newtheorem{proposition}[theorem]{Proposition}

\newtheorem{assumption}[theorem]{Assumption}

\newcommand{\C}{{\mathbb C}}

\newcommand{\be}{\begin{equation}}
\newcommand{\ee}{\end{equation}}
\newcommand{\bea}{\begin{eqnarray}}
\newcommand{\eea}{\end{eqnarray}}
\newcommand{\ba}{\begin{array}}
\newcommand{\ea}{\end{array}}

\newcommand{\ol}{\overline}

\newcommand{\id}{\mathbb{I}}

\newcommand{\re}{\mathrm{Re}}

\newcommand{\eps}{\varepsilon}

\newcommand{\Lam}{\Lambda}
\newcommand{\gam}{\gamma}

\newcommand{\Om}{\Omega}

\newcommand{\dta}{\delta}

\newcommand{\tha}{\theta}

\newcommand{\pt}{\partial}

\numberwithin{equation}{section}

\begin{document}
\title[RHP for the SS equation on the interval]{The unified transform method for the Sasa-Satsuma equation on the interval}

\author[J.Xu]{Jian Xu*}
\address{College of Science\\
University of Shanghai for Science and Technology\\
Shanghai 200093\\
People's  Republic of China}
\email{correspondence author: jianxu@usst.edu.cn}

\author[Q.Zhu]{Qiaozhen Zhu}
\address{School of Mathematical Sciences, Institute of Mathematics and Key Laboratory of Mathematics for Nonlinear Science\\
Fudan University\\
Shanghai 200433\\
People's  Republic of China}
\email{qiaozhenzhu13@fudan.edu.cn}

\author[E.Fan]{Engui Fan}
\address{School of Mathematical Sciences, Institute of Mathematics and Key Laboratory of Mathematics for Nonlinear Science\\
Fudan University\\
Shanghai 200433\\
People's  Republic of China}
\email{faneg@fudan.edu.cn}


\date{\today}

\begin{abstract}
We present a Riemann-Hilbert problem formalism for the initial-boundary value problem for the Sasa-Satsuma(SS) equation on the finite interval. Assume that the solution existes, we show that this solution can be expressed in terms of the solution of a $3\times 3$ Riemann-Hilbert problem. The relevant jump matrices are explicitly given in terms of the three matrix-value spectral functions $s(k)$, $S(k)$ and $S_L(k)$, which in turn are defined in terms of the initial values, boundary values at $x=0$ and boundary values at $x=L$, respectively. However, for a well-posed problem, only part of the boundary values can be prescribed, the remaining boundary data cannot be independently specified, but are determined by the so-called global relation. Here, we analyze the global relation to characterize the unknown boundary values in terms of the given initial and boundary data.

\end{abstract}

\maketitle

\section{Introduction}

Integrable nonlinear partial differential equations (PDEs) can be written as the compatibility condition of two linear eigenvalue equations,
which are called a Lax pair \cite{lax}. They can be analyzed by means of the
Inverse Scattering Transform (IST) formalism which was discovered in 1967 by Gardner, Greene, Kruskal, Miura \cite{ggkm}. Until the 1990s the IST methodology was pursued almost entirely for pure initial value problems. However, in many laboratory and field situations, the wave motion is initiated by what corresponds to the imposition of boundary conditions rather than initial conditions. This naturally leads to the formulation of an initial-boundary value (IBV) problem instead of a pure initial value problem. However, the presence of a boundary presents new challenges.
\par
In 1997, Fokas announced a new unified approach for the analysis of IBV problems for linear and nonlinear integrable PDEs \cite{f1,f2}(see also \cite{f3}). The Fokas method provides a generalization of the IST formalism from initial value to IBV problems. Just like the IST on the line, the Fokas method yields an expression for the solution of an IBV problem in terms of the solution of a Riemann-Hilbert problem.
We know the IST method is usually based on analyzing the $x-part$ of the Lax pairis to get the scattering data. The $t-part$ of the Lax pairs is only used to determine the time evolution of the scattering data. However, the Fokas method is based on the {\it simultaneous} spectral analysis of both $x-part$ and $t-part$ of the Lax pair. Hence, we need all the boundary value data to formulate a Riemann-Hilbert problem. But for a well-posed problem, only part of the boundary values can be prescribed, the remaining boundary data cannot be independently specified, but are determined by the so-called global relation, which is an algebraic relation coupling the relevant spectral functions. Thus, we need analyze this global relation to characterize the unknown boundary value data before we formulate a Riemann-Hilbert problem. This is usually the major difficulty of the Fokas method.

\par
Over the last two decades, the Fokas method was successfully used to analyze boundary value problems for several of the most important integrable equations with $2\times 2$ Lax pairs, such as the Korteweg-de Vries \cite{fi1994}, the nonlinear Schr\"odinger \cite{fis2005}, and the sine-Gordon \cite{fi1992} equations, and so on. However, in \cite{l2012}, the Fokas methodology was further developed to include the case of equations with $3\times 3$ Lax pairs. Since the work of \cite{l2012}, the IBV problems were considered on the half-line. It is a natural problem to extend the Fokas methodology to the case of IBV problems on an finite interval.

\par

In a previous paper \cite{jf3}, we successfully extended the methodology to the case of two-component nonlinear Schr\"odinger equation on the interval. In this paper, we try to analyze the IBV problem for the Sasa-Satsuma equation on the finite interval. We use the idea of choosing the integral contours to determine the appropriate functions to formulate a Riemann-Hilbert problem in \cite{jf3}. However, the most major difficulty and differences from the two-component nonlinear Schr\"odinger equation case is to analyze the global relation to express the unknown boundary value date in terms of the initial value data and the prescribed boundary data in the Sasa-Satsuma equation case. It is not suitable for Sasa-Satsuma equation to do the analysis in spectral domain as \cite{jf3}. However, there is another method to analyze the global relation. Here, we analyse the global relation by using it, i.e. a Gelfand-Levitan-Marchenko(GLM) representation to derive an expression to characterize the unknown boundary values. Since the order of the derivative is higher than the two-component nonlinear Schr\"odinger equation case, the analysis of the global relation of the Sasa-Satsuma equation will be more complicate.

\par
{\bf Organization of the paper:} In section 2 we state the main equation considered in this paper, and perform the spectral analysis of the associated Lax pair. Then, we can formulate the main Riemann-Hilbert problem. It is shown that the solution of the IBV problem for the Sasa-Satsuma equation on the finite interval can be recovered by the solution of this Riemann-Hilbert problem in section 3. The most important result, which expresses the unknown boundary value data in terms of the initial value and prescribed boundary value date by analysing the global relation, is obtained in section 4.

\section{Spectral analysis}

\subsection{Sasa-Satsuma equation}

\par
It is well known that the nonlinear Schr\"odinger(NLS) equation
\be \label{NLSe}
iq_T+\frac{1}{2}q_{XX}+|q|^2q=0
\ee
describes slowly varying wave envelopes in dispersive media and arises in various physical systems such as water waves, plasma physics, solid-state physics and nonlinear optics. One of the most successful among them is the description of optical solitons in fibers. However, by the advancement of experomenal accuracy, several phenomena which can not be explained by equation (\ref{NLSe}) have been observed. In order to understand such phenomena, Kodama and Hasegawa \cite{KH1987} proposed a higer-order nonlinear Schr\"odinger equation
\be \label{HNLSe}
iq_T+\frac{1}{2}q_{XX}+|q|^2q+i\eps \{\beta_1q_{XXX}+\beta_2|q|^2q_{X}+\beta_3q(|q|^2)_X\}=0.
\ee
\par
In general, equation (\ref{HNLSe}) may not be completely integrable. However, if some restrictions are imposed on the real parameters $\beta_1,\beta_2$ and $\beta_3$, it becomes integrable, then we can apply the IST to solve its initial value problems. Until now, the following four cases besides the NLS equation itself are konwn to be integrable:
\begin{itemize}
\item the Chen-Lee-Liu-type derivative NLS equation ($\beta_1:\beta_2:\beta_3$=0:1:1),
\item the Kaup-Newell-type derivative NLS equation ($\beta_1:\beta_2:\beta_3$=0:1:0),
\item the Hirota equation($\beta_1:\beta_2:\beta_3$=1:6:0),
\item the Sasa-Satsuma equation($\beta_1:\beta_2:\beta_3$=1:6:3).
       \be\label{SSe}
       iq_T+\frac{1}{2}q_{XX}+|q|^2q+i\eps (q_{XXX}+6|q|^2q_X+3q(|q|^2)_X)=0
       \ee
\end{itemize}

Let us consider the last case, i.e. $\beta_1:\beta_2:\beta_3$=1:6:3.
According to \cite{ss} we introduce variable transformations,
%\begin{subequations} \label{vartrans}
%\be\label{uqtrans}
\be \label{vartrans}
u(x,t)=q(X,T)\exp\{\frac{-i}{6\eps}(X-\frac{T}{18\eps})\},\quad t=T,\quad x=X-\frac{T}{12\eps}.
\ee
%\be
%t=T,
%\ee
%\be
%x=X-\frac{T}{12\eps}.
%\ee
%\end{subequations}
Then equation (\ref{SSe}) is reduce to a complex modified KdV-type equation
\be\label{KdV-typee}
u_t+\eps\{u_{xxx}+6|u|^2u_x+3u(|u|^2)_x\}=0.
\ee
In the following , we consider the equation (\ref{KdV-typee}) instead of the equation (\ref{SSe}).

\par
The initial value problem for the Sasa-Satsuma equation (\ref{KdV-typee}) on the whole line has been obtained in \cite{ss} by the IST method. And, the IBV problem on the half-line domain was considered in \cite{jfss}. However, in what follows, we consider the IBV problem for the equation (\ref{KdV-typee}) on the interval $\Omega=\{(x,t)|0\le x\le L, 0\le t\le T\}$, where $L$ is a positive finite constant, $T$ is a given finite time, in this paper, that is,
\be\label{IBV}
\ba{lll}
\mbox{Initial value: }&u_0(x)=u(x,0),&\\
\mbox{Dirichlet Boundary value: }& g_0(t)=u(0,t),& f_0(t)=u(L,t),\\
\mbox{First Neumann Boundary value: }& g_1(t)=u_x(0,t),& f_1(t)=u_x(L,t),\\
\mbox{Second Neumann Boundary value: }& g_2(t)=u_{xx}(0,t),& f_2(t)=u_{xx}(L,t).
\ea
\ee

\par

\subsection{Lax pair}

The Lax pair of equation (\ref{KdV-typee}) is (see, \cite{ss}),
%\begin{subequations}
\be\label{Lax-x}
\Psi_x=U\Psi,\quad \Psi_t=V\Psi,\quad \Psi=\left(\ba{ccc}\Psi_1&\Psi_2&\Psi_3\ea\right)^T.
\ee
%\be\label{Lax-t}
%\Psi_t=V\Psi.
%\ee
%\end{subequations}
where
\be\label{Udef}
U=-ik\Lam+V_1,\quad V=-4i\eps k^3\Lam+V_2.
\ee
%and
%\be\label{Vdef}
%V=-4i\eps k^3\Lam+V_2
%\ee
here
\be\label{Lamdef}
\Lam=\left(\ba{ccc}1&0&0\\0&1&0\\0&0&-1\ea\right),V_1=\left(\ba{ccc}0&0&u\\0&0&\bar u\\-\bar u&-u&0\ea\right),V_2=4k^2V_2^{(2)}+2ikV_2^{(1)}+V_2^{(0)}.
\ee
where
\begin{subequations}
\be
\ba{ll}
V_2^{(2)}=\eps\left(\ba{ccc}0&0&u\\0&0&\bar u\\-\bar u&-u&0\ea\right),&
V_2^{(1)}=\eps\left(\ba{ccc}|u|^2&u^2&u_x\\\bar u^2&|u|^2&\bar u_x\\\bar u_x&u_x&-2|u|^2\ea\right),
\ea
\ee
\be
V_2^{(0)}=-4|u|^2\eps\left(\ba{ccc}0&0&u\\0&0&\bar u\\-\bar u&-u&0\ea\right)-\eps\left(\ba{ccc}0&0&u_{xx}\\0&0&\bar u_{xx}\\-\bar u_{xx}&-u_{xx}&0\ea\right)+\eps(u\bar u_x-u_x\bar u)\left(\ba{ccc}1&0&0\\0&-1&0\\0&0&0\ea\right)
%\ea
\ee
\end{subequations}
In the following, we let $\eps=1$ for the convenient of the analysis.

\subsection{The closed one-form}
Suppose that $u(x,t)$ is sufficiently smooth function of $(x,t)$ in the finite interval domain $\Om=\{0<x<L,0<t<T\}$. Introducing a new eigenfunction $\mu(x,t,k)$ by
\be\label{neweigfun}
\Psi(x,t,k)=\mu(x,t,k) e^{-i\Lam kx-4i\Lam k^3t}
\ee
then we find the Lax pair equations
\be\label{muLaxe}
\left\{
\ba{l}
\mu_x+[ik\Lam,\mu]=V_1\mu,\\
\mu_t+[4ik^3\Lam,\mu]=V_2\mu.
\ea
\right.
\ee
Letting $\hat A$ denotes the operators which acts on a $3\times 3$ matrix $X$ by $\hat A X=[A,X]$, the equations in (\ref{muLaxe}) can be written in differential form as
\be\label{mudiffform}
d(e^{(ikx+4ik^3t)\hat \Lam}\mu)=W,
\ee
where $W(x,t,k)$ is the closed one-form defined by
\be\label{Wdef}
W=e^{(ikx+4ik^3t)\hat \Lam}(V_1dx+V_2dt)\mu.
\ee

\subsection{The $\mu_j$'s definition}
We define four eigenfunctions $\{\mu_j\}_1^4$ of (\ref{muLaxe}) by the Volterra integral equations
\be\label{mujdef}
\mu_j(x,t,k)=\id+\int_{\gam_j}e^{(-i kx-4i k^3t)\hat \Lam}W_j(x',t',k).\qquad j=1,2,3,4.
\ee
where $\id$ denotes the $3\times 3$ identity matrix, $W_j$ is given by (\ref{Wdef}) with $\mu$ replaced with $\mu_j$, and the contours $\{\gam_j\}_1^4$ are showed in Figure \ref{fig-1}.
\begin{figure}[th]
\centering
\includegraphics{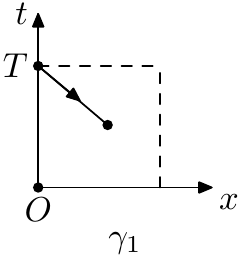}
\includegraphics{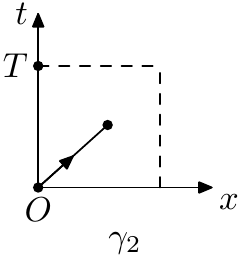}
\includegraphics{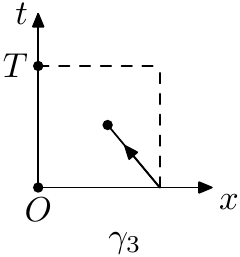}
\includegraphics{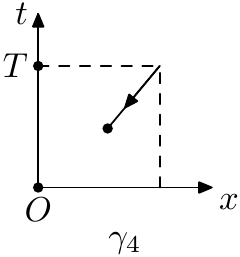}
\caption{The four contours $\gam_1,\gam_2,\gam_3$ and $\gam_4$ in the $(x,t)-$domain.}\label{fig-1}
\end{figure}

The first, second and third column of the matrix equation (\ref{mujdef}) involves the exponentials
\be
\ba{ll}
\mbox{$[\mu_j]_1$:}&e^{2ik(x-x')+8ik^3(t-t')},\\
\mbox{$[\mu_j]_2$:}&e^{2ik(x-x')+8ik^3(t-t')},\\
\mbox{$[\mu_j]_3$:}&e^{-2ik(x-x')-8ik^3(t-t')},e^{-2ik(x-x')-8ik^3(t-t')}.
\ea
\ee
And we have the following inequalities on the contours:
\be
\ba{llll}
\gam_1:&x-x'\ge 0,t-t'\le 0,&
\gam_2:&x-x'\ge 0,t-t'\ge 0,\\
\gam_3:&x-x'\le 0,t-t'\ge 0,&
\gam_4:&x-x'\le 0,t-t'\le 0.
\ea
\ee
So, these inequalities imply that the functions $\{\mu_j\}_1^4$ are bounded and analytic for $k\in\C$ such that $k$ belongs to
\be\label{mujbodanydom}
\ba{llll}
\mu_1:&(D_2,D_2,D_3);&
\mu_2:&(D_1,D_1,D_4);\\
\mu_3:&(D_3,D_3,D_2);&
\mu_4:&(D_4,D_4,D_1),
\ea
\ee
where $\{D_n\}_1^4$ denote four open, pairwisely disjoint subsets of the complex $k-$plane showed in Figure \ref{fig-2}.
\begin{figure}[th]
\centering
\includegraphics{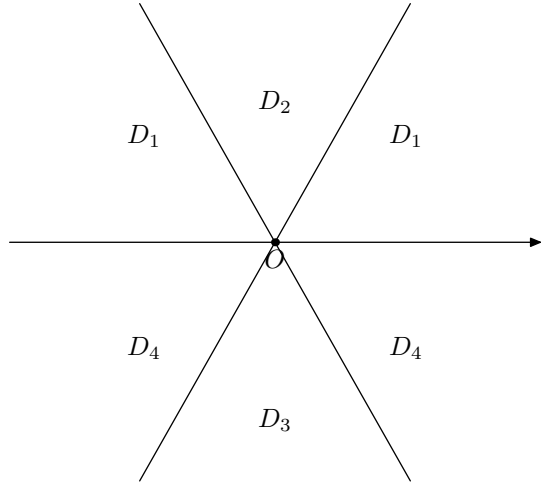}
\caption{The sets $D_n$, $n=1,\ldots ,4$, which decompose the complex $k-$plane.}\label{fig-2}
\end{figure}
And the sets $\{D_n\}_1^4$ has the following properties:
\[
\ba{l}
D_1=\{k\in\C|\re{l_1}=\re{l_2}>\re{l_3},\re{z_1}=\re{z_2}>\re{z_3}\},\\
D_2=\{k\in\C|\re{l_1}=\re{l_2}>\re{l_3},\re{z_1}=\re{z_2}<\re{z_3}\},\\
D_3=\{k\in\C|\re{l_1}=\re{l_2}<\re{l_3},\re{z_1}=\re{z_2}>\re{z_3}\},\\
D_4=\{k\in\C|\re{l_1}=\re{l_2}<\re{l_3},\re{z_1}=\re{z_2}<\re{z_3}\},\\
\ea
\]
where $l_i(k)$ and $z_i(k)$ are the diagonal entries of matrices $-ik\Lam$ and $-4ik^3\Lam$, respectively.

\subsection{The $M_n$'s definition}
For each $n=1,\ldots,4$, define a solution $M_n(x,t,k)$ of (\ref{muLaxe}) by the following system of integral equations:
\be\label{Mndef}
(M_n)_{ij}(x,t,k)=\dta_{ij}+\int_{\gam_{ij}^n}(e^{(-i kx-4i k^3t)\hat \Lam}W_n(x',t',k))_{ij},\quad k\in D_n,\quad i,j=1,2,3.
\ee
where $W_n$ is given by (\ref{Wdef}) with $\mu$ replaced with $M_n$, and the contours $\gam_{ij}^n$, $n=1,\ldots,4$, $i,j=1,2,3$ are defined by
\be\label{gamijndef}
\gam_{ij}^n=\left\{
\ba{lclcl}
\gam_1&if&\re l_i(k)<\re l_j(k)&and&\re z_i(k)\ge\re z_j(k),\\
\gam_2&if&\re l_i(k)<\re l_j(k)&and&\re z_i(k)<\re z_j(k),\\
\gam_3&if&\re l_i(k)\ge\re l_j(k)&and&\re z_i(k)\le \re z_j(k),\\
\gam_4&if&\re l_i(k)\ge\re l_j(k)&and&\re z_i(k)\ge \re z_j(k).
\ea
\right.
\quad \mbox{for }\quad k\in D_n.
\ee
Here, we make a distinction between the contours $\gam_3$ and $\gam_4$ as follows,
\be
\gam^{n}_{ij}=\left\{
\ba{lcl}
\gam_3,&if&\prod_{1\le i<j\le 3} (\re l_i(k)-\re l_j(k))(\re z_i(k)-\re z_j(k))<0,\\
\gam_4,&if&\prod_{1\le i<j\le 3} (\re l_i(k)-\re l_j(k))(\re z_i(k)-\re z_j(k))>0.
\ea
\right.
\ee
The rule chosen in the produce is if $l_m=l_n$, $m$ may not equals $n$, we just choose the subscript is smaller one.
\par

According to the definition of the $\gam^n$, we find that
\be\label{gamndef}
\ba{ll}
\gam^1=\left(\ba{lll}\gam_4&\gam_4&\gam_4\\\gam_4&\gam_4&\gam_4\\\gam_2&\gam_2&\gam_4\ea\right)&
\gam^2=\left(\ba{lll}\gam_3&\gam_3&\gam_3\\\gam_3&\gam_3&\gam_3\\\gam_1&\gam_1&\gam_3\ea\right)\\
\gam^3=\left(\ba{lll}\gam_3&\gam_3&\gam_1\\\gam_3&\gam_3&\gam_1\\\gam_3&\gam_3&\gam_3\ea\right)&
\gam^4=\left(\ba{lll}\gam_4&\gam_4&\gam_2\\\gam_4&\gam_4&\gam_2\\\gam_4&\gam_4&\gam_4\ea\right).
\ea
\ee

\par
The following proposition ascertains that the $M_n$'s defined in this way have the properties required for the formulation of a Riemann-Hilbert problem.
\begin{proposition}
For each $n=1,\ldots,4$, the function $M_n(x,t,k)$ is well-defined by equation (\ref{Mndef}) for $k\in \bar D_n$ and $(x,t)\in \Om$. For any fixed point $(x,t)$, $M_n$ is bounded and analytic as a function of $k\in D_n$ away from a possible discrete set of singularities $\{k_j\}$ at which the Fredholm determinant vanishes. Moreover, $M_n$ admits a bounded and contious extension to $\bar D_n$ and
\be\label{Mnasy}
M_n(x,t,k)=\id+O(\frac{1}{k}),\qquad k\rightarrow \infty,\quad k\in D_n.
\ee
\end{proposition}
\begin{proof}
The bounedness and analyticity properties are established in appendix B in \cite{l2012}. And substituting the expansion
\[
M=M_0+\frac{M^{(1)}}{k}+\frac{M^{(2)}}{k^2}+\cdots,\qquad k\rightarrow \infty.
\]
into the Lax pair (\ref{muLaxe}) and comparing the terms of the same order of $k$ yield the equation (\ref{Mnasy}).
\end{proof}

\subsection{The jump matrices}

We define spectral functions $S_n(k)$, $n=1,\ldots,4$, and
\be\label{Sndef}
S_n(k)=M_n(0,0,k),\qquad k\in D_n,\quad n=1,\ldots,4.
\ee
Let $M$ denote the sectionally analytic function on the complex $k-$plane which equals $M_n$ for $k\in D_n$. Then $M$ satisfies the jump conditions
\be\label{Mjump}
M_n=M_mJ_{m,n},\qquad k\in \bar D_n\cap \bar D_m,\qquad n,m=1,\ldots,4,\quad n\ne m,
\ee
where the jump matrices $J_{m,n}(x,t,k)$ are defined by
\be\label{Jmndef}
J_{m,n}=e^{(-i kx-4i k^3t)\hat \Lam}(S_m^{-1}S_n).
\ee

\subsection{The adjugated eigenfunctions}
We will also need the analyticity and boundedness properties of the minors of the matrices $\{\mu_j(x,t,k)\}_1^3$. We recall that the cofactor matrix $X^A$ of a $3\times 3$ matrix $X$ is defined by
\[
X^A=\left(
\ba{ccc}
m_{11}(X)&-m_{12}(X)&m_{13}(X)\\
-m_{21}(X)&m_{22}(X)&-m_{23}(X)\\
m_{31}(X)&-m_{32}(X)&m_{33}(X)
\ea
\right),
\]
where $m_{ij}(X)$ denote the $(ij)-$th minor of $X$.
\par
It follows from (\ref{muLaxe}) that the adjugated eigenfunction $\mu^A$ satisfies the Lax pair
\be\label{muadgLaxe}
\left\{
\ba{l}
\mu_x^A-[ik\Lam,\mu^A]=-V_1^T\mu^A,\\
\mu_t^A-[4ik^3\Lam,\mu^A]=-V_2^T\mu^A.
\ea
\right.
\ee
where $V^T$ denote the transform of a matrix $V$.
Thus, the eigenfunctions $\{\mu_j^A\}_1^4$ are solutions of the integral equations
\be\label{muadgdef}
\mu_j^A(x,t,k)=\id-\int_{\gam_j}e^{ik(x-x')+4ik^3(t-t')\hat \Lam}(V_1^Tdx+V_2^T)\mu^A,\quad j=1,2,3,4.
\ee
Then we can get the following analyticity and boundedness properties:
\be\label{mujadgbodanydom}
\ba{llll}
\mu_1^A:&(D_3,D_3,D_2);&
\mu_2^A:&(D_4,D_4,D_1);\\
\mu_3^A:&(D_2,D_2,D_3);&
\mu_3^A:&(D_1,D_1,D_4).
\ea
\ee

\subsection{The $J_{m,n}$'s computation}

Let us define the $3\times 3-$matrix value spectral functions $s(k)$, $S(k)$ and $S_L(k)$ by
\begin{subequations}\label{sSdef}
\be\label{mu3mu2s}
\mu_3(x,t,k)=\mu_2(x,t,k)e^{(-ikx-4ik^3t)\hat \Lam}s(k),
\ee
\be\label{mu1mu2S}
\mu_1(x,t,k)=\mu_2(x,t,k)e^{(-ikx-4ik^3t)\hat \Lam}S(k),
\ee
\be\label{mu4mu3SL}
\mu_4(x,t,k)=\mu_3(x,t,k)e^{(-ik(x-L)-4ik^3t)\hat \Lam}S_L(k)
\ee
\end{subequations}
Thus,
\begin{subequations}
\be\label{smu3}
s(k)=\mu_3(0,0,k),
\ee
\be\label{Smu1}
S(k)=\mu_1(0,0,k)=e^{4ik^3 T\hat \Lam}\mu_2^{-1}(0,T,k),
\ee
\be\label{SLmu4}
S_L(k)=\mu_4(L,0,k)=e^{4ik^3 T\hat \Lam}\mu_3^{-1}(L,T,k).
\ee
\end{subequations}
And we deduce from the properties of $\mu_j$ and $\mu_j^A$ that $\{s(k),S(k),S_L(k)\}$ and $\{s^A(k),S^A(k),S^A_L(k)\}$ have the following boundedness properties:
\[
\ba{llll}
s(k):&(D_3\cup D_4,D_3\cup D_4,D_1\cup D_2),&
S(k):&(D_2\cup D_4,D_2\cup D_4,D_1\cup D_3),\\
S_L(k):&(D_2\cup D_4,D_2\cup D_3,D_1\cup D_3)&
s^A(k):&(D_1\cup D_2,D_1\cup D_2,D_3\cup D_4),\\
S^A(k):&(D_1\cup D_3,D_1\cup D_3,D_2\cup D_4),&
S_L^A(k):&(D_1\cup D_3,D_1\cup D_3,D_2\cup D_4).
\ea
\]
Moreover, we notice that
\be\label{MnSnrel}
M_n(x,t,k)=\mu_2(x,t,k)e^{(-ikx-4ik^3t)\hat\Lam}S_n(k),\quad k\in D_n.
\ee
\begin{proposition}
The $S_n$ can be expressed in terms of the entries of $s(k),S(k)$ and $S_L(k)$ as follows:
\begin{subequations}\label{Sn}
\be
\footnotesize
\ba{ll}
S_1=\left(\ba{ccc}\frac{m_{22}(\mathcal{A})}{\mathcal{A}_{33}}&\frac{m_{21}(\mathcal{A})}{\mathcal{A}_{33}}&\mathcal{A}_{13}\\\frac{m_{12}(\mathcal{A})}{\mathcal{A}_{33}}&\frac{m_{11}(\mathcal{A})}{\mathcal{A}_{33}}
&\mathcal{A}_{23}\\0&0&\mathcal{A}_{33}\ea\right),&
S_2=\left(\ba{ccc}\frac{m_{22}(s)m_{33}(S)-m_{32}(s)m_{23}(S)}{(s^TS^A)_{33}}&\frac{m_{21}(s)m_{33}(S)-m_{31}(s)m_{23}(S)}{(s^TS^A)_{33}}&s_{13}\\
\frac{m_{12}(s)m_{33}(S)-m_{32}(s)m_{13}(S)}{(s^TS^A)_{33}}&\frac{m_{11}(s)m_{33}(S)-m_{31}(s)m_{13}(S)}{(s^TS^A)_{33}}&s_{23}\\
\frac{m_{12}(s)m_{23}(S)-m_{22}(s)m_{13}(S)}{(s^TS^A)_{33}}&\frac{m_{11}(s)m_{23}(S)-m_{21}(s)m_{13}(S)}{(s^TS^A)_{33}}&s_{33}\ea\right),\\
\ea
\ee
\be
\small
\ba{ll}
S_3=\left(\ba{ccc}s_{11}&s_{12}&\frac{S_{13}}{(S^Ts^A)_{33}}\\s_{21}&s_{22}&\frac{S_{23}}{(S^Ts^A)_{33}}\\s_{31}&s_{32}&\frac{S_{33}}{(S^Ts^A)_{33}}\ea\right),&
S_4=\left(\ba{ccc}\mathcal{A}_{11}&\mathcal{A}_{12}&0\\\mathcal{A}_{21}&\mathcal{A}_{22}&0\\\mathcal{A}_{31}&\mathcal{A}_{32}&\frac{1}{m_{33}(\mathcal{A})}\ea\right).
\ea
\ee
\end{subequations}
where $\mathcal{A}=(\mathcal{A}_{ij})_{i,j=1}^{3}=s(k)e^{ikL\hat \Lam}S_L(k)$.
\end{proposition}
\begin{proof}
Firstly, we define $R_n(k),T_n(k)$ and $Q_n(k)$ as follows:
\begin{subequations}\label{RnTnQn}
\be\label{Rn}
R_n(k)=e^{4ik^3T\hat \Lam}M_n(0,T,k),
\ee
\be\label{Tn}
T_n(k)=e^{ikL\hat \Lam}M_n(L,0,k),
\ee
\be\label{Qn}
Q_n(k)=e^{(ikL+4ik^3T)\hat\Lam}M_n(L,T,k).
\ee
\end{subequations}
Then, we have the following relations:
\be\label{MnRnSnTn}
\left\{
\ba{l}
M_n(x,t,k)=\mu_1(x,t,k)e^{(-ikx-4ik^3t)\hat\Lam}R_n(k),\\
M_n(x,t,k)=\mu_2(x,t,k)e^{(-ikx-4ik^3t)\hat\Lam}S_n(k),\\
M_n(x,t,k)=\mu_3(x,t,k)e^{(-ikx-4ik^3t)\hat\Lam}T_n(k),\\
M_n(x,t,k)=\mu_3(x,t,k)e^{(-ikx-4ik^3t)\hat\Lam}Q_n(k)
\ea
\right.
\ee

The relations (\ref{MnRnSnTn}) imply that
\be\label{sSRnSnTn}
\ba{lll}
s(k)=S_n(k)T^{-1}_n(k),&
S(k)=S_n(k)R^{-1}_n(k),&
\mathcal{A}(k)=S_n(k)Q^{-1}_n(k).
\ea
\ee
These equations constitute a matrix factorization problem which, given $\{s(k),S(k),S_L(k)\}$ can be solved for the $\{R_n,S_n,T_n,Q_n\}$. Indeed, the integral equations (\ref{Mndef}) together with the definitions of $\{R_n,S_n,T_n,Q_n\}$ imply that
\be
\left\{
\ba{llllll}
(R_n(k))_{ij}=0&if&\gam_{ij}^n=\gam_1;&
(S_n(k))_{ij}=0&if&\gam_{ij}^n=\gam_2;\\
(T_n(k))_{ij}=\dta_{ij}&if&\gam_{ij}^n=\gam_3;&
(Q_n(k))_{ij}=\dta_{ij}&if&\gam_{ij}^n=\gam_4.
\ea
\right.
\ee
It follows that (\ref{sSRnSnTn}) are 27 scalar equations for 27 unknowns. By computing the explicit solution of this algebraic system, we arrive at (\ref{Sn}).
\end{proof}

\subsection{The global relation}
The spectral functions $S(k),S_L(k)$ and $s(k)$ are not independent but satisfy an important relation. Indeed, it follows from (\ref{sSdef}) that
\be
\mu_1(x,t,k)e^{(-ikx-4ik^3t)\hat\Lam}\{S^{-1}(k)s(k)e^{ikL\hat\Lam}S_L(k)\}=\mu_4(x,t,k).%\quad k\in(D_3\cup D_4,D_3\cup D_4,D_1\cup D_2).
\ee
Since $\mu_1(0,T,k)=\id$, evaluation at $(0,T)$ yields the following global relation:
\be\label{globalrel}
S^{-1}(k)s(k)e^{ikL\hat\Lam}S_L(k)=e^{4ik^3T\hat \Lam}c(T,k),%\quad k\in(D_3\cup D_4,D_3\cup D_4,D_1\cup D_2).
\ee
where $c(T,k)=\mu_4(0,T,k)$.

\subsection{The residue conditions}
Since $\mu_2$ is an entire function, it follows from (\ref{MnSnrel}) that M can only have singularities at the points where the $S_n's$ have singularities. We infer from the explicit formulas (\ref{Sn}) that the possible singularities of $M$ are as follows:
\begin{itemize}
\item $[M]_1$ could have poles in $D_1$ at the zeros of $\mathcal{A}_{33}(k)$;
\item $[M]_1$ could have poles in $D_2$ at the zeros of $(s^TS^A)_{33}(k)$;
\item $[M]_2$ could have poles in $D_1$ at the zeros of $\mathcal{A}_{33}(k)$;
\item $[M]_2$ could have poles in $D_2$ at the zeros of $(s^TS^A)_{33}(k)$;
\item $[M]_3$ could have poles in $D_3$ at the zeros of $(S^Ts^A)_{33}(k)$;
\item $[M]_3$ could have poles in $D_4$ at the zeros of $m_{33}(\mathcal{A})(k)$;
\end{itemize}
We denote the above possible zeros by $\{k_j\}_1^N$ and assume they satisfy the following assumption.
\begin{assumption}\label{zeroassump}
We assume that
\begin{itemize}
\item $\mathcal{A}_{33}(k)$ has $n_0$ possible simple zeros in $D_1$ denoted by $\{k_j\}_1^{n_0}$;
\item $(s^TS^A)_{33}(k)$ has $n_1-n_0$ possible simple zeros in $D_2$ denoted by $\{k_j\}_{n_0+1}^{n_1}$;
\item $(S^Ts^A)_{33}(k)$ has $n_2-n_1$ possible simple zeros in $D_3$ denoted by $\{k_j\}_{n_1+1}^{n_2}$;
\item $m_{33}(\mathcal{A})(k)$ has $N-n_2$ possible simple zeros in $D_4$ denoted by $\{k_j\}_{n_2+1}^{N}$;
\end{itemize}
and that none of these zeros coincide. Moreover, we assume that none of these functions have zeros on the boundaries of the $D_n$'s.
\end{assumption}
We determine the residue conditions at these zeros in the following:
\begin{proposition}\label{propos}
Let $\{M_n\}_1^4$ be the eigenfunctions defined by (\ref{Mndef}) and assume that the set $\{k_j\}_1^N$ of singularitues are as the above assumption. Then the following residue conditions hold:
\begin{subequations}
\be\label{M11D1res}
{Res}_{k=k_j}[M]_1=\frac{m_{12}(\mathcal{A})(k_j)}{\dot{\mathcal{A}}_{33}(k_j)\mathcal{A}_{23}(k_j)}e^{\tha_{31}(k_j)}[M(k_j)]_3,\quad 1\le j\le n_0,k_j\in D_1
\ee
\be\label{M12D1res}
{Res}_{k=k_j}[M]_2=\frac{m_{12}(\mathcal{A})(k_j)}{\dot{\mathcal{A}}_{33}(k_j)\mathcal{A}_{13}(k_j)}e^{\tha_{32}(k_j)}[M(k_j)]_3,\quad 1\le j\le n_0,k_j\in D_1
\ee
\be\label{M21D2res}
\ba{r}
Res_{k=k_j}[M]_1=\frac{m_{12}(s)(k_j)m_{33}(S)(k_j)-m_{32}(s)(k_j)m_{13}(S)(k_j)}{\dot{(s^TS^A)_{33}(k_j)}s_{23}(k_j)}e^{\tha_{31}(k_j)}[M(k_j)]_3\\
\quad n_0+1\le j\le n_1,k_j\in D_2,
\ea
\ee
\be\label{M22D2res}
\ba{r}
Res_{k=k_j}[M]_2=\frac{m_{21}(s)(k_j)m_{33}(S)(k_j)-m_{31}(s)(k_j)m_{23}(S)(k_j)}{\dot{(s^TS^A)_{33}(k_j)}s_{13}(k_j)}e^{\tha_{32}(k_j)}[M(k_j)]_3\\
\quad n_0+1\le j\le n_1,k_j\in D_2,
\ea
\ee
\be\label{M43D4res}
\ba{rl}
Res_{k=k_j}[M]_3=&\frac{S_{13}(k_j)s_{32}(k_j)-S_{33}(k_j)s_{12}(k_j)}{\dot{(S^Ts^A)_{33}(k_j)}m_{23}(s)(k_j)}e^{\tha_{13}(k_j)}[M(k_j)]_1\\
&+\frac{S_{33}(k_j)s_{11}(k_j)-S_{13}(k_j)s_{31}(k_j)}{\dot{(S^Ts^A)_{33}(k_j)}m_{23}(s)(k_j)}e^{\tha_{23}(k_j)}[M(k_j)]_2,n_1+1\le j\le n_2,k_j\in D_3,
\ea
\ee
\be\label{M63D6res}
\ba{r}
Res_{k=k_j}[M]_3=\frac{\mathcal{A}_{12}(k_j)}{\dot m_{33}(\mathcal{A})(k_j)m_{23}(\mathcal{A})(k_j)}e^{\tha_{13}(k_j)}[M(k_j)]_1-\frac{\mathcal{A}_{11}(k_j)}{\dot m_{33}(\mathcal{A})(k_j)m_{23}(\mathcal{A})(k_j)}e^{\tha_{23}(k_j)}[M(k_j)]_2\\
\quad n_2+1\le j\le N,k_j\in D_4.
\ea
\ee
\end{subequations}
where $\dot f=\frac{df}{dk}$, and $\tha_{ij}$ is defined by
\be\label{thaijdef}
\tha_{ij}(x,t,k)=(l_i-l_j)x+(z_i-z_j)t,\quad i,j=1,2,3.
\ee
that implies that
\[
\tha_{ij}=0,i,j=1,2;\quad \tha_{13}=\tha_{23}=-\tha_{32}=-\tha_{31}=-2ikx-8ik^3t.
\]
\end{proposition}
\begin{proof}
We will prove (\ref{M11D1res}), (\ref{M21D2res}), (\ref{M43D4res}), (\ref{M63D6res}), the other conditions follow by similar arguments.
Equation (\ref{MnSnrel}) implies the relation
\begin{subequations}
\be\label{M1S1}
M_1=\mu_2e^{(-ikx-4ik^3t)\hat\Lam}S_1,
\ee
\be\label{M2S2}
M_2=\mu_2e^{(-ikx-4ik^3t)\hat\Lam}S_2.
\ee
\be\label{M4S4}
M_3=\mu_2e^{(-ikx-4ik^3t)\hat\Lam}S_3,
\ee
\be\label{M6S6}
M_4=\mu_2e^{(-ikx-4ik^3t)\hat\Lam}S_4,
\ee
\end{subequations}
In view of the expressions for $S_1$ and $S_2$ given in (\ref{Sn}), the three columns of (\ref{M1S1}) read:
\begin{subequations}
\be\label{M11}
[M_1]_1=[\mu_2]_1\frac{m_{22}(\mathcal{A})}{\mathcal{A}_{33}}+[\mu_2]_2e^{\tha_{21}}\frac{m_{12}(\mathcal{A})}{\mathcal{A}_{33}},
\ee
\be\label{M12}
[M_1]_2=[\mu_2]_1e^{\tha_{12}}\frac{m_{21}(\mathcal{A})}{\mathcal{A}_{33}}+[\mu_2]_2\frac{m_{11}(\mathcal{A})}{\mathcal{A}_{33}},
\ee
\be\label{M13}
[M_1]_3=[\mu_2]_1e^{\tha_{13}}\mathcal{A}_{13}+[\mu_2]_2e^{\tha_{23}}\mathcal{A}_{23}+[\mu_2]_3\mathcal{A}_{33}.
\ee
\end{subequations}
while the three columns of (\ref{M2S2}) read:
\begin{subequations}
\be\label{M21}
\ba{rl}
[M_2]_1&=[\mu_2]_1\frac{m_{22}(s)m_{33}(S)-m_{32}(s)m_{23}(S)}{(s^TS^A)_{33}}\\
&+[\mu_2]_2\frac{m_{12}(s)m_{33}(S)-m_{32}(s)m_{13}(S)}{(s^TS^A)_{33}}e^{\tha_{21}}\\
&+[\mu_2]_3\frac{m_{12}(s)m_{23}(S)-m_{22}(s)m_{13}(S)}{(s^TS^A)_{33}}e^{\tha_{31}}
\ea
\ee
\be\label{M22}
\ba{rl}
[M_2]_2&=[\mu_2]_1\frac{m_{21}(s)m_{33}(S)-m_{31}(s)m_{23}(S)}{(s^TS^A)_{33}}e^{\tha_{12}}\\
&+[\mu_2]_2\frac{m_{11}(s)m_{33}(S)-m_{31}(s)m_{13}(S)}{(s^TS^A)_{33}}\\
&+[\mu_2]_3\frac{m_{11}(s)m_{23}(S)-m_{21}(s)m_{13}(S)}{(s^TS^A)_{33}}e^{\tha_{32}}
\ea
\ee
\be\label{M23}
[M_2]_3=[\mu_2]_1s_{13}e^{\tha_{13}}+[\mu_2]_2s_{23}e^{\tha_{23}}+[\mu_2]_3s_{33}.
\ee
\end{subequations}
and the three columns of (\ref{M4S4}) read:
\begin{subequations}
\be\label{M41}
[M_3]_1=[\mu_2]_1s_{11}+[\mu_2]_2s_{21}e^{\tha_{21}}+[\mu_2]_3s_{31}e^{\tha_{31}},
\ee
\be\label{M42}
[M_3]_2=[\mu_2]_1s_{12}e^{\tha_{12}}+[\mu_2]_2s_{22}+[\mu_2]_3s_{32}e^{\tha_{32}},
\ee
\be\label{M43}
[M_3]_3=[\mu_2]_1\frac{S_{13}}{(S^Ts^A)_{33}}e^{\tha_{13}}+[\mu_2]_2\frac{S_{23}}{(S^Ts^A)_{33}}e^{\tha_{23}}+[\mu_2]_3\frac{S_{33}}{(S^Ts^A)_{33}}.
\ee
\end{subequations}
the three columns of (\ref{M6S6}) read:
\begin{subequations}
\be\label{M61}
[M_4]_1=[\mu_2]_1\mathcal{A}_{11}+[\mu_2]_2\mathcal{A}_{21}e^{\tha_{21}}+[\mu_2]_3\mathcal{A}_{31}e^{\tha_{31}},
\ee
\be\label{M62}
[M_4]_2=[\mu_2]_1\mathcal{A}_{12}e^{\tha_{12}}+[\mu_2]_2\mathcal{A}_{22}+[\mu_2]_3\mathcal{A}_{32}e^{\tha_{32}},
\ee
\be\label{M63}
[M_4]_3=[\mu_2]_3\frac{1}{m_{33}(\mathcal{A})}.
\ee
\end{subequations}
We first suppose that $k_j\in D_1$ is a simple zero of $s_{33}(k)$. Solving (\ref{M13}) for $[\mu_2]_2$ and substituting the result in to (\ref{M11}), we find
\[
[M_1]_1=\frac{m_{12}(\mathcal{A})}{\mathcal{A}_{33}\mathcal{A}_{23}}e^{\tha_{31}}[M_1]_3+\frac{m_{32}(\mathcal{A})}{\mathcal{A}_{23}}[\mu_2]_2-\frac{m_{12}(\mathcal{A})}{\mathcal{A}_{23}}e^{\tha_{31}}[\mu_2]_3.
\]
Taking the residue of this equation at $k_j$, we find the condition (\ref{M11D1res}) in the case when $k_j\in D_1$. Similarly, Solving (\ref{M23}) for $[\mu_2]_2$ and substituting the result in to (\ref{M21}), we find
\[
[M_2]_1=\frac{m_{12}(s)m_{33}(S)-m_{32}(s)m_{13}(S)}{(s^TS^A)_{33}s_{23}}e^{\tha_{31}}[M_1]_3-\frac{m_{32}(s)}{s_{23}}[\mu_2]_1-\frac{m_{12}(s)}{s_{23}}e^{\tha_{31}}[\mu_2]_3.
\]
Taking the residue of this equation at $k_j$, we find the condition (\ref{M21D2res}) in the case when $k_j\in D_2$.
\par
In order to prove (\ref{M43D4res}), we solve (\ref{M41}) and (\ref{M42}) for $[\mu_2]_1$ and $[\mu_2]_3$, then substituting the result into (\ref{M43}), we find
\[
[M_3]_3=\frac{S_{13}s_{32}-S_{33}s_{12}}{(S^Ts^A)_{33}m_{23}(s)}e^{\tha_{13}}[M_3]_1+\frac{S_{33}s_{11}-S_{13}s_{31}}{(S^Ts^A)_{33}m_{23}(s)}e^{\tha_{23}}[M_3]_2+\frac{1}{m_{23}(s)}[\mu_2]_3.
\]
Taking the residue of this equation at $k_j$, we find the condition (\ref{M43D4res}) in the case when $k_j\in D_3$. Similarly, solving (\ref{M61}) and (\ref{M62}) for $[\mu_2]_1$ and $[\mu_2]_3$, then substituting the result into (\ref{M63}), we find
\[
[M_4]_3=\frac{\mathcal{A}_{12}}{m_{33}(\mathcal{A})m_{23}(\mathcal{A})}e^{\tha_{13}}[M_4]_1-\frac{\mathcal{A}_{11}}{m_{33}(\mathcal{A})m_{23}(\mathcal{A})}e^{\tha_{13}}[M_4]_2-\frac{1}{m_{23}(\mathcal{A})}e^{\tha_{23}}[\mu_2]_2.
\]
Taking the residue of this equation at $k_j$, we find the condition (\ref{M63D6res}) in the case when $k_j\in D_4$.
\end{proof}

\section{The Riemann-Hilbert problem}

The sectionally analytic function $M(x,t,k)$ defined in section 2 satisfies a Riemann-Hilbert problem which can be formulated in terms of the initial and boundary values of $u(x,t)$. By solving this Riemann-Hilbert problem, the solution of (\ref{KdV-typee})(then (\ref{SSe})) can be recovered for all values of $x,t$.
\begin{theorem}
Suppose that $u(x,t)$ is a solution of (\ref{KdV-typee}) in the interval domain $\Om$ with sufficient smoothness. Then $u(x,t)$ can be reconstructed from the initial value $\{u_0(x)\}$ and boundary values $\{g_0(t),g_1(t),g_2(t)\}$, $\{f_0(t),f_1(t),f_2(t)\}$ defined as (\ref{IBV}).
\par
Use the initial and boundary data to define the jump matrices $J_{m,n}(x,t,k)$ as well as the spectral $s(k)$ and $S(k)$ by equation (\ref{sSdef}). Assume that the possible zeros $\{k_j\}_1^N$ of the functions $\mathcal{A}_{33}(k),(s^TS^A)_{33}(k),(S^Ts^A)_{33}(k)$ and $m_{33}(\mathcal{A})(k)$ are as in assumption \ref{zeroassump}.
\par
Then the solution $\{u(x,t)\}$ is given by
\be\label{usolRHP}
u(x,t)=2i\lim_{k\rightarrow \infty}(kM(x,t,k))_{13}.
\ee
where $M(x,t,k)$ satisfies the following $3\times 3$ matrix Riemann-Hilbert problem:
\begin{itemize}
\item $M$ is sectionally meromorphic on the Riemann $k-$sphere with jumps across the contours $\bar D_n\cap \bar D_m,n,m=1,\cdots, 4$, see Figure \ref{fig-2}.
\item Across the contours $\bar D_n\cap \bar D_m$, $M$ satisfies the jump condition
      \be\label{MRHP}
      M_n(x,t,k)=M_m(x,t,k)J_{m,n}(x,t,k),\quad k\in \bar D_n\cap \bar D_m,n,m=1,2,3,4.
      \ee
      where the jump matrices $J_{m,n}(x,t,k)$ are defined as (\ref{Jmndef}).
\item $M(x,t,k)=\id+O(\frac{1}{k}),\qquad k\rightarrow \infty$.
\item The residue condition of $M$ is showed in Proposition \ref{propos}.
\end{itemize}
\end{theorem}
\begin{proof}
It only remains to prove (\ref{usolRHP}) and this equation follows from the large $k$ asymptotics of the eigenfunctions.
\end{proof}

\section{The solution of the global relation}
A major difficulty of initial-boundary value problems is that some of the boundary values are unknown for a well-posed problem. All boundary values are needed for the definition of $S(k),S_L(k)$, and hence for the formulation of the Riemann-Hilbert problem. In this section, our main result, theorem 4.2, expresses the spectral function $S(k),S_L(k)$ in terms of the prescribed boundary data and the initial data via the solution of a system of nonlinear integral equations.

We define functions $\{\Phi_{ij}(t,k)\}_{i,j=1}^{3}$ and $\{\phi_{ij}(t,k)\}_{i,j=1}^{3}$ by:
\be
\footnotesize
\ba{ll}
\mu_2(0,t,k)=\left(\ba{lll}\Phi_{11}(t,k)&\Phi_{12}(t,k)&\Phi_{13}(t,k)\\
\Phi_{21}(t,k)&\Phi_{22}(t,k)&\Phi_{23}(t,k)\\\Phi_{31}(t,k)&\Phi_{32}(t,k)&\Phi_{33}(t,k)\ea\right),
&
\mu_3(L,t,k)=\left(\ba{lll}\phi_{11}(t,k)&\phi_{12}(t,k)&\phi_{13}(t,k)\\
\phi_{21}(t,k)&\phi_{22}(t,k)&\phi_{23}(t,k)\\\phi_{31}(t,k)&\phi_{32}(t,k)&\phi_{33}(t,k)\ea\right).
\ea
\ee

\par
We will first derive Gelfand-Levitan-Marchenko (GLM) representations for the eigenfunctions $\Phi_{ij}$ and $\phi_{ij}$, then consider the solution of the global relation, which leads to expressions for the unknown boundary values from the known ones in terms of the GLM representations.

\subsection{The GLM representation}

\begin{theorem}
The eigenfunctions $\Phi_{ij}$ and $\phi_{ij}$ admit the following GLM representations,
%\begin{subequations}\label{PhiGLM}
\small
\be\label{PhiGLM}
\Phi_{ij}(t,k)=\dta_{ij}+\int_{-t}^{t}\left(\left[\left(\tilde{L}-\frac{1}{2}V^{(2)}_{2}\Lam \tilde{M}+\frac{1}{4}V^{(2)}_{2x}N\right)+ik\left(\tilde{M}+\frac{1}{2}V^{(2)}_{2}\Lam N\right)+k^2N
\right]e^{4ik^3(t-s)\Lam}\right)_{ij}ds
\ee
%\end{subequations}
%\begin{subequations}\label{phiGLM}
\small
\be\label{phiGLM}
\phi_{ij}(t,k)=\dta_{ij}+\int_{-t}^{t}\left(\left[\left(\tilde{\mathcal{L}}-\frac{1}{2}V^{(2)}_{2}\Lam \tilde{\mathcal{M}}+\frac{1}{4}V^{(2)}_{2x}\mathcal{N}\right)+ik\left(\tilde{\mathcal{M}}+\frac{1}{2}V^{(2)}_{2}\Lam \mathcal{N}\right)+k^2\mathcal{N}
\right]e^{4ik^3(t-s)\Lam}\right)_{ij}ds
\ee
%\end{subequations}
\normalsize
where $\dta_{ij}=\left\{\ba{ll}1,&i=j\\0,&i\ne j\ea\right.$ and the $3\times 3$ matrices $\tilde L(t,s),\tilde M(t,s)$ and $N(t,s)$ satisfy the initial conditions
\begin{subequations}\label{GLMint}
\be
\left\{
\ba{l}
N(t,-t)+\Lam N(t,-t)\Lam=0,\\
\tilde{M}(t,-t)+\Lam \tilde{M}(t,-t)\Lam=0,\\
\tilde{L}(t,-t)+\Lam \tilde{L}(t,-t)\Lam=0.
\ea
\right.
\ee
\be
\left\{
\ba{l}
N(t,t)-\Lam N(t,t)\Lam=4V^{(2)}_{2},\\
\tilde{M}(t,t)-\Lam \tilde{M}(t,t)\Lam=2\Lam V^{(2)}_{2x},\\
\tilde{L}(t,t)-\Lam \tilde{L}(t,t)\Lam=2(V^{(2)}_{2})^3-V^{(2)}_{2xx}.
\ea
\right.
\ee
\end{subequations}
and an ODE systems
\begin{subequations}\label{GLMsys}
\be
\left\{
\ba{l}
N_{t}+\Lam N_{s} \Lam =\left[(V^{(2)}_{2})^3)-V^{(2)}_{2xx}\right]N-2\Lam V^{(2)}_{2x}\tilde{M}+4V^{(2)}_{2}\tilde{L},\\
\tilde{M}_{t}+\Lam \tilde{M}_{s} \Lam=\left[(V^{(2)}_{2})^3-V^{(2)}_{2xx}\right]\tilde{M}+{\bf A}\Lam N+2\Lam V^{(2)}_{2x}\tilde{L},\\
\tilde{L}_{t}+\Lam \tilde{L}_{s} \Lam=\left[2(V^{(2)}_{2})^3-V^{(2)}_{2xx}\right]\tilde{L}+{\bf B}\Lam \tilde{M}+{\bf D} N
\ea
\right.
\ee
where the matrices ${\bf A,B,D}$ are defined as
\begin{subequations}
\be
\ba{rcl}
{\bf A}&=&\frac{3}{2}(V^{(2)}_{2})^4-\frac{1}{2}(V^{(2)}_{2}V^{(2)}_{2xx}+V^{(2)}_{2xx}V^{(2)}_{2})+\frac{1}{2}(V^{(2)}_{2x})^2\\
&&+\frac{1}{2}[(V^{(2)}_{2})^2V^{(2)}_{2x}+V^{(2)}_{2x}(V^{(2)}_{2})^2-V^{(2)}_{2}V^{(2)}_{2x}V^{(2)}_{2}]-\frac{1}{2}\dot{V}^{(2)}_{2}
\ea
\ee
\be
{\bf B}=\frac{1}{2}(V^{(2)}_{2}V^{(2)}_{2xx}+V^{(2)}_{2xx}V^{(2)}_{2})-\frac{1}{2}(V^{(2)}_{2x})^{2}-\frac{3}{2}(V^{(2)}_{2})^4+\frac{1}{2}\dot{V}^{(2)}_{2}-\frac{1}{2}[V^{(2)}_{2x}(V^{(2)}_{2})^2-V^{(2)}_{2}V^{(2)}_{2x}V^{(2)}_{2}]
\ee
\be
\ba{rcl}
{\bf D}&=&-\frac{1}{4}\dot{V}^{(2)}_{2x}+\frac{3}{4}(V^{(2)}_{2})^5+\frac{1}{4}V^{(2)}_{2x}V^{(2)}_{2}V^{(2)}_{2x}-\frac{1}{4}((V^{(2)}_{2})^2V^{(2)}_{2xx}+V^{(2)}_{2}V^{(2)}_{2xx}V^{(2)}_{2})\\
&&+\frac{1}{4}V^{(2)}_{2}\dot{V}^{(2)}_{2}+\frac{1}{4}(V^{(2)}_{2x}V^{(2)}_{2xx}-V^{(2)}_{2xx}V^{(2)}_{2x})+\frac{1}{4}((V^{(2)}_{2})^3V^{(2)}_{2x}-V^{(2)}_{2x}(V^{(2)}_{2})^3)
\ea
\ee
\end{subequations}
\end{subequations}

Analogously, the functions $\{\tilde{\mathcal{L}}(t,s),\tilde{\mathcal{M}}(t,s),\mathcal{N}(t,s)\}$ satisfy the similar system of equations with $\{g_0(t),g_1(t),g_2(t)\}$ replaced by $\{f_0(t),f_1(t),f_2(t)\}$.
\end{theorem}

\begin{proof}
We just prove the GLM representations (\ref{PhiGLM}) of the eigenfunctions $\Phi_{ij}$, the GLM representations (\ref{phiGLM}) is similar. We substitute the following equation
\be
 \mu(t,k)=\id+\int_{-t}^{t}\left(L(t,s)+ik M(t,s)+k^2 N(t,s)\right)e^{4ik^3(t-s)\Lam}ds
\ee
into the $t-$part of the Lax pair (\ref{muLaxe}), then, we can find the matrices $L(t,s),M(t,s)$ and $N(t,s)$ are satisfied
\begin{subequations}
\be
\left\{
\ba{l}
N(t,-t)+\Lam N(t,-t)\Lam =0,\\
M(t,-t)+\Lam M(t,-t)\Lam +V^{(2)}_{2}N(t,-t)\Lam=0,\\
L(t,-t)+\Lam L(t,-t)\Lam -V^{(2)}_{2}M(t,-t)\Lam -\frac{1}{2}V^{(1)}_{2}N(t,-t)\Lam=0.
\ea
\right.
\ee
\be
\left\{
\ba{l}
N(t,t)-\Lam N(t,t)\Lam=4V^{(2)}_{2},\\
{M}(t,t)-\Lam {M}(t,t)\Lam=2\Lam V^{(1)}_{2}+V^{(2)}_{2}N(t,t)\Lam,\\
{L}(t,t)-\Lam {L}(t,t)\Lam=V^{(0)}_{2}-V^{(2)}_{2}M(t,t)\Lam-\frac{1}{2}V^{(1)}_{2}N(t,t)\Lam.
\ea
\right.
\ee
\be
\left\{
\ba{l}
N_{t}(t,s)+\Lam N_{s}(t,s)\Lam=4V^{(2)}_{2}L(t,s)-2V^{(1)}_{2}M(t,s)+V^{(0)}_{2}N(t,s),\\
M_{t}(t,s)+\Lam M_{s}(t,s)\Lam=2V^{(1)}_{2}L(t,s)+V^{(0)}_{2}M(t,s)-V^{(2)}_{2}N_{s}(t,s)\Lam,\\
L_{t}(t,s)+\Lam L_{s}(t,s)\Lam=V^{(0)}_{2}L(t,s)+V^{(2)}_{2}M_{s}(t,s)\Lam+\frac{1}{2}V^{(1)}_{2}N_{s}(t,s)\Lam,
\ea
\right.
\ee
\end{subequations}
If we introducing the following transformations,
\be\label{LMtoLMtilde}
\ba{l}
M(t,s)=\tilde{M}(t,s)+\frac{1}{2}V^{(2)}_{2}\Lam N(t,s),\\
L(t,s)=\tilde{L}(t,s)-\frac{1}{2}V^{(2)}_{2}\Lam \tilde{M}(t,s)+\frac{1}{4}V^{(2)}_{2x}N(t,s).
\ea
\ee
We can get the initial conditions (\ref{GLMint}) and the ODE systems (\ref{GLMsys}).
\end{proof}

\subsection{The Analysis of the global relation}
To avoid routine technical complications, we will continue our analysis of the global relation in the case $s(k)=\id$ corresponding to the zero initial conditions $u_0(x)=u(x,0)=0$. In this case, the global relation (\ref{globalrel}) takes the form
\be\label{zeroglobal}
%\left(e^{4ik^3t\hat\Lam}\Phi(t,k)\right)e^{ikL\hat\Lam}\left(e^{4ik^3t\hat\Lam}\ol{\phi^{T}(t,\bar k)}\right)=e^{4ik^3t\hat\Lam}c(t,k).
\Phi(t,k)e^{ikL\hat\Lam}\left(\ol{\phi(t,\bar k)}\right)^{T}=c(t,k)
\ee
The $(1,3)th$ term of (\ref{zeroglobal}) is
\be\label{global13th}
\Phi_{11}(t,k)\ol{\phi_{31}(t,\bar k)}e^{2ikL}+\Phi_{12}(t,k)\ol{\phi_{32}(t,\bar k)}e^{2ikL}+\Phi_{13}(t,k)\ol{\phi_{33}(t,\bar k)}=c_{13}(t,k),\quad k\in D_1\cup D_3,
\ee
where $c_{13}(t,k)$ has the following properties which are important for the analysis of the global relation
\be\label{c13proper}
\ba{ll}
c_{13}(t,k)=O\left(\frac{1}{k}\right),& k\in D_{1},\\
e^{-2ikL}c_{13}(t,k)=O\left(\frac{1}{k}\right),& k\in D_{3}.
\ea
\ee

In the following, the given boundary conditions are $g_0(t),f_0(t)$ and $f_1(t)$, and we are looking for expressions for $g_1(t),g_2(t)$ and $f_2(t)$.

\par
Let us give some notation.  Define $\omega=e^{\frac{2\pi}{3}}$ and suppose that $\Pi(t,k)$ is a scalar function.
\begin{itemize}\label{notation}
  \item Let $\hat{\Pi}(t,k)$ and denotes the following notation:
    \be
    \ba{l}
    \hat{\Pi}(t,k)=\Pi(t,k)+\omega \Pi(t,\omega k)+\omega^2 \Pi(t,\omega^2 k),\\
     \widehat{\ol{\Pi}}(t,\bar k)=\ol{\Pi}(t,\bar k)+\omega \ol{\Pi}(t,\ol{\omega k})+\omega^2 \ol{\Pi}(t,\ol{\omega^2 k})
     \ea
    \ee
  \item Let $\tilde{\Pi}(t,k)$ denotes the following notation:
    \be
    \ba{l}
    \tilde{\Pi}(t,k)=\Pi(t,k)+\omega^2 \Pi(t,\omega k)+\omega \Pi(t,\omega^2 k),\\
     \tilde{\ol{\Pi}}(t,\bar k)=\ol{\Pi}(t,\bar k)+\omega^2 \ol{\Pi}(t,\ol{\omega k})+\omega \ol{\Pi}(t,\ol{\omega^2 k})
     \ea
    \ee

\end{itemize}

\begin{theorem}
Consider the IBV problem for the equation (\ref{KdV-typee}) on the interval.
Given boundary conditions $g_0(t),f_0(t)$ and $f_1(t)$, and we have the expressions for $g_1(t),g_2(t)$ and $f_2(t)$ as follows:
\begin{subequations}\label{g12f2}
\be
\ba{rcl}
g_1(t)&=&\frac{g_0(t)}{\pi}\int_{\pt D^0}\hat{\Phi}_{33}(t,k)dk+\frac{6}{\pi}\int_{\pt D^{0}}(-ik)[\Gamma^{-1}(k)]_3\left(\ba{c}\tilde{G}_1(t,k)\\\tilde{G}_1(t,\omega k)\\e^{-2i\omega^2 kL}\tilde{G}_1(t,\omega^2 k)\ea\right)dk\\
&&{}{}+\frac{6}{\pi}\int_{\pt D^{0}}(-ik)[\Gamma^{-1}(k)]_3\left(\ba{c}G_2(t,k)\\G_2(t,\omega k)\\e^{-2i\omega^2 kL}G_2(t,\omega^2 k)\ea\right)e^{8ik^3(t-t')}dk
\ea
\ee
\be
\ba{rcl}
g_2(t)&=&-2|g_0(t)|^2g_0(t)-\frac{2ig_0(t)}{\pi}\int_{\pt D^0}\tilde{\Phi}_{33}(t,k)dk+\frac{g_1(t)}{\pi}\int_{\pt D^0}\hat{\Phi}_{33}(t,k)dk\\
&&{}+\frac{6}{\pi}\int_{\pt D^{0}}(k^2)[\Gamma^{-1}(k)]_2\left(\ba{c}\tilde{G}_1(t,k)\\\tilde{G}_1(t,\omega k)\\e^{-2i\omega^2 kL}\tilde{G}_1(t,\omega^2 k)\ea\right)dk\\
&&{}+\frac{6}{\pi}\int_{\pt D^{0}}(k^2)[\Gamma^{-1}(k)]_2\left(\ba{c}G_2(t,k)\\G_2(t,\omega k)\\e^{-2i\omega^2 kL}G_2(t,\omega^2 k)\ea\right)e^{8ik^3(t-t')}dk
\ea
\ee
\be
\ba{rcl}
f_2(t)&=&-2|f_0(t)|^2f_0(t)-\frac{2i}{\pi}\left[f_0(t)\int_{\pt D^0}k\tilde{\ol{\phi}}_{11}(t,\bar k)dk+\bar{f}_0(t)\int_{\pt D^0}k\tilde{\ol{\phi}}_{21}(t,\bar k)dk\right]\\
&&{}+\frac{1}{\pi}\int_{\pt D^0}\left[f_1(t)\widehat{\ol{\phi}}_{11}(t,\bar k)dk+\bar{f_1}(t)\widehat{\ol{\phi}}_{21}(t,\bar k)\right]dk\\
&&+\frac{6}{\pi}\int_{\pt D^{0}}(k^2)[\Gamma^{-1}(k)]_1\left(\ba{c}\tilde{G}_1(t,k)\\\tilde{G}_1(t,\omega k)\\e^{-2i\omega^2 kL}\tilde{G}_1(t,\omega^2 k)\ea\right)dk\\
&&{}+\frac{6}{\pi}\int_{\pt D^{0}}(k^2)[\Gamma^{-1}(k)]_1\left(\ba{c}G_2(t,k)\\G_2(t,\omega k)\\e^{-2i\omega^2 kL}G_2(t,\omega^2 k)\ea\right)e^{8ik^3(t-t')}dk
\ea
\ee
where
\be
\ba{rl}
\tilde{G}_1(t,k)=&e^{2ikL}\left\{-\frac{1}{3}\tilde{\ol{\phi}}_{31}(t,\bar k)-\frac{1}{4k^2}\left[f_1(t)-\frac{1}{\pi}\int_{\pt D^0}\left(f_0(t)\widehat{\ol{\phi}}_{11}(t,\bar k)+\bar{f}_0(t)\widehat{\ol{\phi}}_{21}(t,\bar k)\right)dk\right]\right\}\\
&{}-e^{2ikL}\left[\frac{1}{3}\widehat{\ol{\phi}}_{31}(t,\bar k)+\frac{f_0(t)}{2ik}\right]-\left[\frac{1}{3}\hat{\Phi}_{13}(t,k)-\frac{g_0(t)}{2ik}\right].
\ea
\ee
\be%\label{G2def}
G_2(t,k)=\left[(\Phi_{11}(t,k)-1)\ol{\phi}_{31}(t,\bar k)+\Phi_{12}(t,k)\ol{\phi}_{32}(t,\bar k)\right]-\Phi_{13}(t,k)\left[\ol{\phi}_{33}(t,\bar k)-1\right],
\ee
and $[\Gamma^{-1}(k)]_j,j=1,2,3,$ denotes the $j-th$ row of the inverse matrix of the $\Gamma(k)$, which is defined as (\ref{Gammak}).
\end{subequations}

%\be
%\footnotesize
%\ba{rl}
%\tilde{G}_1(t,k)=&2e^{2ikL}\left\{ik\int_{0}^{t}\left[\ol{\tilde{\mathcal{M}}_{31}}(t,2\tau-t)-\frac{1}{2}\left[f_0(t)\ol{\mathcal{N}}_{11}(t,2\tau-t)+\bar{f}_0(t)
%\ol{\mathcal{N}}_{21}(t,2\tau-t)\right]\right]e^{8ik^3(\tau-t)}d\tau\right.\\
%&{}\left.-\frac{1}{8k^2}\left[f_1(t)-\frac{1}{2}\left[f_0(t)\ol{\mathcal{N}}_{11}(t,t)+\bar{f}_0(t)
%\ol{\mathcal{N}}_{21}(t,t)\right]\right]\right\}\\
%&{}-2e^{2ikL}\left\{k^2\int_{0}^{t}\ol{\mathcal{N}}_{31}(t,2\tau-t)e^{8ik^3(\tau-t)}d\tau+\frac{1}{4ik}f_0(t)\right\}\\
%&{}-2\left\{k^2\int_{0}^{t}N_{13}(t,2\tau-t)e^{8ik^3(\tau-t)}d\tau-\frac{1}{4ik}g_0(t)\right\}.
%\ea
%\ee
%\par
%The functions $\tilde{G}_1(t,k),G_2(t,k),N_{33}(t,t),\tilde{M}_{33}(t,t),\ol{\tilde{\mathcal{M}}}_{11}(t,t),\ol{\tilde{\mathcal{M}}}_{21}(t,t),$ and $\ol{\mathcal{N}}_{11}(t,t),\ol{\mathcal{N}}_{21}(t,t)$ involved in (\ref{g12f2}) can be expressed in terms of $\Phi_{ij}(t,k)$ and $\phi_{ij}(t,k)$.
\end{theorem}

\begin{proof}
Substitute the GLM representations of $\Phi_{ij}$ and $\phi_{ij}$ into (\ref{global13th}) and rewrite the resulting equation in the form:
\be\label{global13thre}
%\small
\ba{l}
e^{2ikL}\int_{-t}^{t}\ol{\mathcal{L}_{31}}(t,s)e^{4ik^3(s-t)}ds+\int_{-t}^{t}L_{13}(t,s)e^{4ik^3(s-t)}ds+ik\int_{-t}^{t}M_{13}(t,s)e^{4ik^3(s-t)}ds\\
=G_{1}(t,k)+G_{2}(t,k)+c_13(t,k)\\
%+c_{13}(t,k).
\ea
\ee
where
\begin{subequations}
%G_1 define
\be
\ba{rl}
G_1(t,k)=&e^{2ikL}\left\{ik\int_{-t}^{t}\ol{\mathcal{M}_{31}}(t,s)e^{4ik^3(s-t)}ds-k^2\int_{-t}^{t}\ol{\mathcal{N}_{31}}(t,s)e^{4ik^3(s-t)}ds\right\}\\
&{}-k^2\int_{-t}^{t}N_{13}(t,s)e^{4ik^3(s-t)}ds
\ea
\ee
%G_2 define
\be
\ba{rl}
G_2(t,k)=&-e^{2ikL}\left\{\int_{-t}^{t}\left[L_{11}(t,s)+ik M_{11}(t,s)+k^2N_{11}(t,s)\right]e^{4ik^3(t-s)}ds\right\}\times\\
&\left\{\int_{-t}^{t}\left[\ol{\mathcal{L}_{31}}(t,s)-ik\ol{\mathcal{M}_{31}}(t,s)+k^2\ol{\mathcal{N}_{31}}(t,s)\right]e^{4ik^3(s-t)}ds\right\}\\
&-e^{2ikL}\left\{\int_{-t}^{t}\left[L_{12}(t,s)+ik M_{12}(t,s)+k^2N_{12}(t,s)\right]e^{4ik^3(t-s)}ds\right\}\times \\ &\left\{\int_{-t}^{t}\left[\ol{\mathcal{L}_{32}}(t,s)-ik\ol{\mathcal{M}_{32}}(t,s)+k^2\ol{\mathcal{N}_{32}}(t,s)\right]e^{4ik^3(s-t)}ds\right\}\\
&-\left\{\int_{-t}^{t}\left[L_{13}(t,s)+ik M_{13}(t,s)+k^2N_{13}(t,s)\right]e^{4ik^3(s-t)}ds\right\}\times \\ &\left\{\int_{-t}^{t}\left[\ol{\mathcal{L}_{33}}(t,s)-ik\ol{\mathcal{M}_{33}}(t,s)+k^2\ol{\mathcal{N}_{33}}(t,s)\right]e^{4ik^3(t-s)}ds\right\}
\ea
\ee
\end{subequations}

Let $D=\{k\left|0<\arg{k}<\frac{\pi}{3}\right.\}$. Consider (\ref{global13thre}) in $D$ as well as replacing $k$ by $\omega k$ and by $\omega^2 k$ in (\ref{global13thre}), we obtain three equations, which are valid for $k\in D$. These equations can be written in the vector form as follows:
\be\label{globalvcetor}
\Gamma(k)U(t,k)=H_1(t,k)+H_2(t,k)+H_c(t,k),\quad k\in D,
\ee
where
\be\label{Gammak}
\ba{ll}
\Gamma(k)=\left(\ba{ccc}e^{2ikL}&1&1\\e^{2i\omega kL}&1&\omega\\1&e^{-2i\omega^2 kL}&\omega^2e^{-2i\omega^2 kL}\ea\right),& H_j(t,k)=\left(\ba{c}G_j(t,k)\\G_j(t,\omega k)\\e^{-2i\omega^2 kL}G_j(t,\omega^2 k)\ea\right),j=1,2,\\
U(t,k)=\left(\ba{c}\int_{-t}^{t}\ol{\mathcal{L}_{31}}(t,s)e^{4ik^3(s-t)}ds\\\int_{-t}^{t}L_{13}(t,s)e^{4ik^3(s-t)}ds\\ik\int_{-t}^{t}M_{13}(t,s)e^{4ik^3(s-t)}ds\ea\right)
&H_c(t,k)=\left(\ba{c}c_{13}(t,k)\\c_{13}(t,\omega k)\\e^{-2i\omega^2 kL}c_{13}(t,\omega^2 k)\ea\right).
\ea
\ee
Notice that $\det{\Gamma(k)}\rightarrow \omega-1\ne 0$ as $|k|\rightarrow \infty,k\in \bar D$.

\par
Multiply (\ref{globalvcetor}) by $\left(\ba{ccc}k^2&0&0\\0&k^2&0\\0&0&-ik\ea\right)\Gamma^{-1}(k)e^{8ik^3(t-t')},0<t'<t$, and integrate along the contour $\pt D^{0}$, which is the boundary of $D$ deformed to pass above the zeros of the $\det{\Gamma(k)}$. Then (\ref{c13proper}) implies that the terms containing $H_c$ vanishes, by Jordan's lemma.

\par
In order to evaluate the other terms we will use the following identities (see, e.g. \cite{afscmp} or \cite{fl3}):
\be\label{changeintegral}
\int_{-t}^{t}f(t,s)e^{4ik^3(s-t)}ds=2\int_{0}^{t}f(t,2\tau-t)e^{8ik^3(\tau-t)}ds
\ee
where $f(t,s)$ is an arbitrary function such that the integral is well-defined, and
\begin{subequations}\label{globaliden}
\be\label{k2iden}
\int_{\pt D^{0}}k^2\int_{0}^{t}\alpha(\tau)e^{8ik^3(t-t')}d\tau dk=\frac{\pi}{12}\alpha(t'),
\ee
\be\label{k34iden}
\int_{\pt D^{0}}k^m\int_{0}^{t}\alpha(\tau)e^{8ik^3(t-t')}d\tau dk=\int_{\pt D^{0}}k^m\left(\int_{0}^{t'}\alpha(\tau)e^{8ik^3(t-t')}d\tau-\frac{1}{8ik^3}\alpha(t')\right)dk
\ee
where $m=3,4$ and $\alpha(\tau)$ is a smooth function for $0<\tau<t$. Then the integration by parts together with Jordan's lemma show that one can pass to the limit as $t'\rightarrow t$ in the right-hand side of (\ref{k34iden}).
\end{subequations}

\par
Applying (\ref{k34iden}) to the integral term containing $H_1$ one obtains
\be
\small
\ba{l}
\int_{\pt D^{0}}{\left(\ba{ccc}k^2&0&0\\0&k^2&0\\0&0&-ik\ea\right)\Gamma^{-1}(k)H_1(t,k)e^{8ik^3(t-t')}dk}\\
{}{}=\int_{\pt D^{0}}{\left(\ba{ccc}k^2&0&0\\0&k^2&0\\0&0&-ik\ea\right)\Gamma^{-1}(k)\left(\ba{c}\tilde{G}_1(t,t',k)\\\tilde{G}_1(t,t',\omega k)\\e^{-2i\omega^2 kL}\tilde{G}_1(t,t',\omega^2 k)\ea\right)dk}
\ea
\ee
where
\be\label{G1tild}
\footnotesize
\ba{rl}
\tilde{G}_1(t,t',k)=&2e^{2ikL}\left\{ik\int_{0}^{t'}\left[\ol{\tilde{\mathcal{M}}_{31}}(t,2\tau-t)-\frac{1}{2}\left[f_0(t)\ol{\mathcal{N}}_{11}(t,2\tau-t)+\bar{f}_0(t)
\ol{\mathcal{N}}_{21}(t,2\tau-t)\right]\right]e^{8ik^3(\tau-t')}d\tau\right.\\
&{}\left.-\frac{1}{8k^2}\left[\ol{\tilde{\mathcal{M}}_{31}}(t,2t'-t)-\frac{1}{2}\left[f_0(t)\ol{\mathcal{N}}_{11}(t,2t'-t)+\bar{f}_0(t)
\ol{\mathcal{N}}_{21}(t,2t'-t)\right]\right]\right\}\\
&{}-2e^{2ikL}\left\{k^2\int_{0}^{t'}\ol{\mathcal{N}}_{31}(t,2\tau-t)e^{8ik^3(\tau-t')}d\tau-\frac{1}{8ik}\ol{\mathcal{N}}_{31}(t,2t'-t)\right\}\\
&{}-2\left\{k^2\int_{0}^{t'}N_{13}(t,2\tau-t)e^{8ik^3(\tau-t')}d\tau-\frac{1}{8ik}N_{13}(t,2t'-t)\right\}.
\ea
\ee
Applying (\ref{k2iden}) to the integral in the left-hand side of (\ref{globalvcetor}) and using the equations (\ref{LMtoLMtilde}) we arrive at the equation
\be
\scriptsize
\ba{l}
\left(\ba{c}\ol{\tilde{\mathcal{L}}}_{31}(t,2t'-t)+\frac{1}{2}\left[f_0(t)\ol{\tilde{\mathcal{M}}}_{11}(t,2t'-t)+\bar{f}_0(t)\ol{\tilde{\mathcal{M}}}_{21}(t,2t'-t)\right]-\frac{1}{4}\left[f_1(t)\ol{\mathcal{N}}_{11}(t,2t'-t)+\bar{f}_1(t)\ol{\mathcal{N}}_{21}(t,2t'-t)\right]\\
\tilde{L}_{13}(t,2t'-t)+\frac{1}{2}g_0(t)\tilde{M}_{33}(t,2t'-t)+\frac{1}{4}g_1(t)N_{33}(t,2t'-t)\\
\tilde{M}_{13}(t,2t'-t)-\frac{1}{2}g_0(t)N_{33}(t,2t'-t)
\ea\right)\\
{}=\frac{6}{\pi}\int_{\pt D^{0}}\left(\ba{ccc}k^2&0&0\\0&k^2&0\\0&0&-ik\ea\right)\Gamma^{-1}(k)\left(\ba{c}\tilde{G}_1(t,t',k)\\\tilde{G}_1(t,t',\omega k)\\e^{-2i\omega^2 kL}\tilde{G}_1(t,t',\omega^2 k)\ea\right)dk\\
{}{}{}+\frac{6}{\pi}\int_{\pt D^{0}}\left(\ba{ccc}k^2&0&0\\0&k^2&0\\0&0&-ik\ea\right)\Gamma^{-1}(k)\left(\ba{c}G_2(t,k)\\G_2(t,\omega k)\\e^{-2i\omega^2 kL}G_2(t,\omega^2 k)\ea\right)e^{8ik^3(t-t')}dk
\ea
\ee
Evaluating this equation at $t'=t$ and using the initial conditions (\ref{GLMint}) we find the following equations for $g_1(t),g_2(t)$ and $f_2(t)$.
\begin{subequations}\label{g12f2}
\be
\ba{rcl}
g_1(t)&=&\frac{1}{2}g_0(t)N_{33}(t,t)+\frac{6}{\pi}\int_{\pt D^{0}}(-ik)[\Gamma^{-1}(k)]_3\left(\ba{c}\tilde{G}_1(t,k)\\\tilde{G}_1(t,\omega k)\\e^{-2i\omega^2 kL}\tilde{G}_1(t,\omega^2 k)\ea\right)dk\\
&&{}{}+\frac{6}{\pi}\int_{\pt D^{0}}(-ik)[\Gamma^{-1}(k)]_3\left(\ba{c}G_2(t,k)\\G_2(t,\omega k)\\e^{-2i\omega^2 kL}G_2(t,\omega^2 k)\ea\right)e^{8ik^3(t-t')}dk
\ea
\ee
\be
\ba{rcl}
g_2(t)&=&-4|g_0(t)|^2g_0(t)+g_0(t)\tilde{M}_{33}(t,t)+\frac{1}{2}g_1(t)N_{33}(t,t)\\
&&{}+\frac{6}{\pi}\int_{\pt D^{0}}(k^2)[\Gamma^{-1}(k)]_2\left(\ba{c}\tilde{G}_1(t,k)\\\tilde{G}_1(t,\omega k)\\e^{-2i\omega^2 kL}\tilde{G}_1(t,\omega^2 k)\ea\right)dk\\
&&{}+\frac{6}{\pi}\int_{\pt D^{0}}(k^2)[\Gamma^{-1}(k)]_2\left(\ba{c}G_2(t,k)\\G_2(t,\omega k)\\e^{-2i\omega^2 kL}G_2(t,\omega^2 k)\ea\right)e^{8ik^3(t-t')}dk
\ea
\ee
\be
\ba{rcl}
f_2(t)&=&-4|f_0(t)|^2f_0(t)-\left[f_0(t)\ol{\tilde{\mathcal{M}}}_{11}(t,t)+\bar{f_0}(t)\ol{\tilde{\mathcal{M}}}_{21}(t,t)\right]\\
&&{}+\frac{1}{2}\left[f_1(t)\ol{\mathcal{N}}_{11}(t,t)+\bar{f_1}(t)\ol{\mathcal{N}}_{21}(t,t)\right]\\
&&+\frac{6}{\pi}\int_{\pt D^{0}}(k^2)[\Gamma^{-1}(k)]_1\left(\ba{c}\tilde{G}_1(t,k)\\\tilde{G}_1(t,\omega k)\\e^{-2i\omega^2 kL}\tilde{G}_1(t,\omega^2 k)\ea\right)dk\\
&&{}+\frac{6}{\pi}\int_{\pt D^{0}}(k^2)[\Gamma^{-1}(k)]_1\left(\ba{c}G_2(t,k)\\G_2(t,\omega k)\\e^{-2i\omega^2 kL}G_2(t,\omega^2 k)\ea\right)e^{8ik^3(t-t')}dk
\ea
\ee
where $[\Gamma^{-1}(k)]_j,j=1,2,3,$ denotes the $j-th$ row of $\Gamma^{-1}(k)$ and
\end{subequations}
\be
\footnotesize
\ba{rl}
\tilde{G}_1(t,k)=&2e^{2ikL}\left\{ik\int_{0}^{t}\left[\ol{\tilde{\mathcal{M}}_{31}}(t,2\tau-t)-\frac{1}{2}\left[f_0(t)\ol{\mathcal{N}}_{11}(t,2\tau-t)+\bar{f}_0(t)
\ol{\mathcal{N}}_{21}(t,2\tau-t)\right]\right]e^{8ik^3(\tau-t)}d\tau\right.\\
&{}\left.-\frac{1}{8k^2}\left[f_1(t)-\frac{1}{2}\left[f_0(t)\ol{\mathcal{N}}_{11}(t,t)+\bar{f}_0(t)
\ol{\mathcal{N}}_{21}(t,t)\right]\right]\right\}\\
&{}-2e^{2ikL}\left\{k^2\int_{0}^{t}\ol{\mathcal{N}}_{31}(t,2\tau-t)e^{8ik^3(\tau-t)}d\tau+\frac{1}{4ik}f_0(t)\right\}\\
&{}-2\left\{k^2\int_{0}^{t}N_{13}(t,2\tau-t)e^{8ik^3(\tau-t)}d\tau-\frac{1}{4ik}g_0(t)\right\}.
\ea
\ee
\par
The functions $\tilde{G}_1(t,k),G_2(t,k),N_{33}(t,t),\tilde{M}_{33}(t,t),\ol{\tilde{\mathcal{M}}}_{11}(t,t),\ol{\tilde{\mathcal{M}}}_{21}(t,t),$ and $\ol{\mathcal{N}}_{11}(t,t),\ol{\mathcal{N}}_{21}(t,t)$ involved in (\ref{g12f2}) can be expressed in terms of $\Phi_{ij}(t,k)$ and $\phi_{ij}(t,k)$.

\par
Indeed, $G_2(t,k)$ can be written as follows:
\be\label{G2def}
G_2(t,k)=\left[(\Phi_{11}(t,k)-1)\ol{\phi}_{31}(t,\bar k)+\Phi_{12}(t,k)\ol{\phi}_{32}(t,\bar k)\right]-\Phi_{13}(t,k)\left[\ol{\phi}_{33}(t,\bar k)-1\right],
\ee
%compute the MN functions
and we recall that
\be\label{Phi33}
\Phi_{33}(t,k)=1+\int_{-t}^{t}\left[L_{33}(t,s)+ik M_{33}(t,s)+k^2N_{33}(t,s)\right]e^{4ik^3(s-t)}ds,
\ee
and then we have
\begin{subequations}
\be
3k^2\int_{-t}^{t}N_{33}(t,s)e^{4ik^3(s-t)}ds=\Phi_{33}(t,k)+\omega \Phi_{33}(t,\omega k)+\omega^2\Phi_{33}(t,\omega^2 k),
\ee
\be
3ik\int_{-t}^{t}M_{33}(t,s)e^{4ik^3(s-t)}ds=\Phi_{33}(t,k)+\omega^2 \Phi_{33}(t,\omega k)+\omega\Phi_{33}(t,\omega^2 k).
\ee
\end{subequations}
Integrating along the contour $\pt D^{0}$ and applying (\ref{changeintegral}) and (\ref{k2iden}) to the above equations we arrive at the equations
\begin{subequations}\label{N33M33tilde}
\be
N_{33}(t,t)=\frac{2}{\pi}\int_{\pt D^{0}}\left[\Phi_{33}(t,k)+\omega \Phi_{33}(t,\omega k)+\omega^2\Phi_{33}(t,\omega^2 k)\right]dk,
\ee
\be
\tilde{M}_{33}(t,t)=2|g_0(t)|^2-\frac{2i}{\pi}\int_{\pt D^{0}}\left[\Phi_{33}(t,k)+\omega^2 \Phi_{33}(t,\omega k)+\omega\Phi_{33}(t,\omega^2 k)\right]dk.
\ee
\end{subequations}
Similarly, we have
\begin{subequations}\label{N11N21tilde}
\be
\ol{\mathcal{N}}_{11}(t,t)=\frac{2}{\pi}\int_{\pt D^{0}}\left[\ol{\phi}_{11}(t,\bar k)+\omega \ol{\phi}_{11}(t,\ol{\omega k})+\omega^2\ol{\phi}_{11}(t,\ol{\omega^2 k})\right]dk,
\ee
\be
\ol{\mathcal{N}}_{21}(t,t)=\frac{2}{\pi}\int_{\pt D^{0}}\left[\ol{\phi}_{21}(t,\bar k)+\omega \ol{\phi}_{21}(t,\ol{\omega k})+\omega^2\ol{\phi}_{21}(t,\ol{\omega^2 k})\right]dk.
\ee
\end{subequations}
\begin{subequations}\label{M11M21tilde}
\be
\ol{\tilde{\mathcal{M}}}_{11}(t,t)=-|f_0(t)|^2+\frac{2i}{\pi}\int_{\pt D^{0}}k\left[\ol{\phi}_{11}(t,\bar k)+\omega^2 \ol{\phi}_{11}(t,\ol{\omega k})+\omega \ol{\phi}_{11}(t,\ol{\omega^2 k})\right]dk,
\ee
\be
\ol{\tilde{\mathcal{M}}}_{21}(t,t)=-f^2_0(t)+\frac{2i}{\pi}\int_{\pt D^{0}}k\left[\ol{\phi}_{21}(t,\bar k)+\omega^2 \ol{\phi}_{21}(t,\ol{\omega k})+\omega \ol{\phi}_{21}(t,\ol{\omega^2 k})\right]dk.
\ee
\end{subequations}

\par
From the equation
\be
\ol{\phi}_{31}(t,\bar k)=\int_{-t}^{t}\left[\ol{\mathcal{L}}_{31}(t,s)-ik\ol{\mathcal{M}}_{31}(t,s)+k^2\ol{\mathcal{N}}_{31}(t,s)\right]e^{4ik^3(s-t)}ds,
\ee
we have
\begin{subequations}
\be
3k^2\int_{-t}^{t}\ol{\mathcal{N}}_{31}(t,s)e^{4ik^3(s-t)}ds=\ol{\phi}_{31}(t,\bar k)+\omega \ol{\phi}_{31}(t,\ol{\omega k})+\omega^2 \ol{\phi}_{31}(t,\ol{\omega^2 k}),
\ee
\be
-3ik\int_{-t}^{t}\ol{\mathcal{M}}_{31}(t,s)e^{4ik^3(s-t)}ds=\ol{\phi}_{31}(t,\bar k)+\omega^2 \ol{\phi}_{31}(t,\ol{\omega k})+\omega \ol{\phi}_{31}(t,\ol{\omega^2 k}),
\ee
\end{subequations}
and these imply
\begin{subequations}
\be
k^2\int_{0}^{t}\ol{\mathcal{N}}_{31}(t,2\tau-t)e^{8ik^3(\tau-t)}d\tau=\frac{1}{6}\left[\ol{\phi}_{31}(t,\bar k)+\omega \ol{\phi}_{31}(t,\ol{\omega k})+\omega^2 \ol{\phi}_{31}(t,\ol{\omega^2 k})\right],
\ee
\be
ik\int_{0}^{t}\ol{\mathcal{M}}_{31}(t,2\tau-t)e^{8ik^3(\tau-t)}d\tau=-\frac{1}{6}\left[\ol{\phi}_{31}(t,\bar k)+\omega^2 \ol{\phi}_{31}(t,\ol{\omega k})+\omega \ol{\phi}_{31}(t,\ol{\omega^2 k})\right].
\ee
\end{subequations}
Notice that
\be
\ol{\mathcal{M}}_{31}(t,s)=\ol{\tilde{\mathcal{M}}}_{31}(t,s)-\frac{1}{2}\left[f_0(t)\ol{\mathcal{N}}_{11}(t,s)+\bar{f}_0(t)\ol{\mathcal{N}}_{21}(t,s)\right].
\ee

From
\be
\Phi_{13}(t,k)=\int_{-t}^{t}\left[L_{13}(t,s)+ikM_{13}(t,s)+k^2N_{13}(t,k)\right]e^{4ik^3(s-t)}ds,
\ee
we have
\be
3k^2\int_{-t}^{t}N_{13}(t,s)e^{4ik^3(s-t)}ds=\Phi_{13}(t,k)+\omega \Phi_{13}(t,\omega k)+\omega^2\Phi_{13}(t,\omega^2k),
\ee
it implies that
\be
k^2\int_{0}^{t}N_{13}(t,2\tau-t)e^{8ik^3(\tau-t)}d\tau=\frac{1}{6}\left[\Phi_{13}(t,k)+\omega \Phi_{13}(t,\omega k)+\omega^2\Phi_{13}(t,\omega^2k)\right].
\ee

\par
Hence,
\be\label{Gtilde1}
\ba{rl}
\tilde{G}_1(t,k)=&e^{2ikL}\left\{-\frac{1}{3}\left[\ol{\phi}_{31}(t,\bar k)+\omega^2 \ol{\phi}_{31}(t,\ol{\omega k})+\omega \ol{\phi}_{31}(t,\ol{\omega^2 k})\right]\right.\\
&{}-\frac{1}{4k^2}\left[f_1(t)-\frac{1}{\pi}f_0(t)\int_{\pt D^{0}}\left[\ol{\phi}_{11}(t,\bar k)+\omega \ol{\phi}_{11}(t,\ol{\omega k})+\omega^2\ol{\phi}_{11}(t,\ol{\omega^2 k})\right]dk\right.\\
&\left.\left.-\bar{f}_0(t)\frac{1}{\pi}\int_{\pt D^{0}}\left[\ol{\phi}_{21}(t,\bar k)+\omega \ol{\phi}_{21}(t,\ol{\omega k})+\omega^2\ol{\phi}_{21}(t,\ol{\omega^2 k})\right]dk\right]\right\}\\
&-e^{2ikL}\left\{\frac{1}{3}\left[\ol{\phi}_{31}(t,\bar k)+\omega \ol{\phi}_{31}(t,\ol{\omega k})+\omega^2 \ol{\phi}_{31}(t,\ol{\omega^2 k})\right]+\frac{1}{2ik}f_0(t)\right\}\\
&-\left\{\frac{1}{3}\left[\Phi_{13}(t,k)+\omega \Phi_{13}(t,\omega k)+\omega^2\Phi_{13}(t,\omega^2k)\right]-\frac{1}{2ik}g_0(t)\right\}.
\ea
\ee
Noticing that our notations \ref{notation} and using (\ref{N33M33tilde}), (\ref{N11N21tilde}), (\ref{M11M21tilde}), (\ref{Gtilde1}) and (\ref{G2def}) in (\ref{g12f2}) we obtain the equations for $g_1(t),g_2(t)$ and $f_2(t)$ in terms of $\Phi_{ij}(t,k)$ and $\phi_{ij}(t,k)$ as (\ref{g12f2}). These equations, together with (\ref{PhiGLM}) and the similar equation (\ref{phiGLM}) constitute a system of nonlinear ODEs for $\Phi_{ij}$ and $\phi_{ij}$.

\end{proof}

%\bigskip

{\bf Acknowledgements}
This work of Xu was supported by National Natural
Science Foundation of China under project NO.11501365, Shanghai Sailing Program
supported by Science and Technology Commission of Shanghai Municipality
under Grant NO.15YF1408100, Shanghai youth teacher assistance program NO.ZZslg15056 and the Hujiang Foundation of China (B14005). Fan was support by grants from the National Natural
Science Foundation of China (Project No.10971031; 11271079; 11075055).


\begin{thebibliography}{XXXX}


\bibitem{lax} P.D. Lax, Integrals of nonlinear equations of evolution and solitary waves, Commun. Pure Appl.
Math. 21(1968) 467-490.

\bibitem{ggkm} C.S. Gardner, J.M. Greene, M.D. Kruskal, R.M. Miura, {\em Methods for solving the Korteweg-de
Vries equation}, Phys. Rev. Lett. {\bf 19}(1967), 1095-1097.

\bibitem{f1} A. S. Fokas, {\em A unified transform method for solving linear and certain nonlinear PDEs}, Proc. R. Soc. Lond. A {\bf
    453}(1997), 1411-1443.

\bibitem{f2} A. S. Fokas, {\em Integrable nonlinear evolution equations on the half-line}, Commun. Math. Phys. {\bf 230}(2002), 1-39.

\bibitem{f3} A.S. Fokas, {\em A Unified Approach to Boundary Value Problems}, in: CBMS-NSF Regional Conference Series in Applied Mathematics, SIAM, 2008.

\bibitem{fi1994} A.S. Fokas, A.R. Its, An initial-boundary value problem for the Korteweg-de Vries equation. Math.
Comput. Simul. 37(1994), 293-321.

\bibitem{fis2005} A.S. Fokas, A.R. Its, L.Y. Sung, The nonlinear Schr\"odinger equation on the half-line. Nonlinearity
18(2005), 1771-1822.

\bibitem{fi1992} A.S. Fokas, A.R. Its, An initial-boundary value problem for the sine-Gordon equation. Theor.Math.
Physics 92(1992), 388-403.


\bibitem{l2012} J. Lenells, {\em Initial-boundary value problems for integrable evolution equations with $3\times 3$ Lax pairs}, Physica D {\bf 241}(2012) 857-875.

\bibitem{jf3} J. Xu, E. Fan, Initial-boundary value problem for integrable nonlinear evolution equations with $3\times 3$ Lax pairs on the interval, to apperear in Stud. Appl. Math.


\bibitem{KH1987} Y. Kodama and A. Hasegawa, Nonlinear pulse propagation in a monomode dielectric guide, Quantum Electronics, IEEE Journal of,  {\bf 23}(1987), 510-524.


\bibitem{ss} N. Sasa, J. Satsuma, {\em New-type of soliton solutions for a higher-order nonlinear Schrödinger equation}, J. Phys. Soc. Japan {\bf 60} (1991) 409?17.


\bibitem{jfss} J. Xu, E. Fan, The unified transform method for the Sasa-Satsuma equation on the half-line, Proc. R. Soc. A. 469(2013) 20130068.

\bibitem{afscmp} A. Boutet de Monvel, A. S. Fokas and D. Shepelsky, {\em Integrable nonlinear evolution equations on a finite interval}, Commun. Math. Phys. {\bf 263}(2006), 133-172.




%\bibitem{abmfs1} A. Boutet De Monvel, A.S. Fokas, D. Shepelsky, {\em Integrable nonlinear evolution equations on a finite interval}, Comm. Math. Phys. {\bf 263} (2006) 133-172.
%
%\bibitem{abmfs2} A. Boutet de Monvel,A.S.Fokas,D.Shepelsky, {\em The mKDV equation on the half-line}, J. Inst. Math. Jussieu.{\bf 3}(2004), 139-164.
%
%\bibitem{fis} A. S. Fokas, A. R. Its and L. Y. Sung, {\em The nonlinear Schr\"odinger equation on the half-line}, Nonlinearity. {\bf  18}(2005), 1771-1822.

%\bibitem{k} S. Kamvissis, {\em Semiclassical nonlinear Schrödinger on the half line}, J. Math. Phys. {\bf 44} (2003) 5849?868.
%
%\bibitem{l1} J. Lenells, {\em Boundary value problems for the stationary axisymmetric Einstein equations: a disk rotating around a black hole}, Comm. Math. Phys. {\bf 304} (2011) 585-635.
%\bibitem{l2} J. Lenells, A.S. Fokas, {\em Boundary-value problems for the stationary axisymmetric Einstein equations: a rotating disc}, Nonlinearity {\bf 24} (2011) 177-206.



%\bibitem{k} D.J. Kaup, {\em On the inverse scattering problem for cubic eigenvalue problems of the class $\psi_{xxx}+6Q\psi_x+6R\psi=\lam\psi$}, Stud. Appl. Math. {\bf 62} (1980) 189-216.
%
%
%
%\bibitem{l4} J. Lenells, {\em The Degasperis-Procesi equation on the half-line}, Nonlinear Analysis {\bf 76}(2013) 122-139.
%
%\bibitem{bc} R. Beals and R. Coifman, {\em Scattering and inverse scattering for first order systems},  Comm. in Pure and Applied Math. {\bf 37}(1984), 39--90.
%
%\bibitem{dz} P. Deift and  X. Zhou, {\em A steepest descent method for oscillatory Riemann--Hilbert problems},  Ann. of Math. (2) {\bf 137}(1993), 295-368.
%
%\bibitem{pdl} P. D. Lax, {\em Integrals of nonlinear equations of evolution and solitary waves}, Comm. Pure. Appl. Math.{\bf 21}(1968), 467-490.
%
%\bibitem{fl1} A. S. Forkas and J. Lenells, {\em The unified method: \Rmnum{1}.nonlinearizable problem on the half-line}, J. Phys. A: Math. Theor. {\bf 45}(2012) 195201;
%
%\bibitem{fl2} J. Lenells and A. S. Forkas, {\em The unified method: \Rmnum{2}. NLS on the half-line t-periodic boundary conditions}, J. Phys. A: Math. Theor. {\bf 45}(2012) 195202;

\bibitem{fl3} J. Lenells and A. S. Forkas, {\em The unified method: \Rmnum{3}. Nonlinearizable problem on the interval}, J. Phys. A: Math. Theor. {\bf 45}(2012) 195203;

\end{thebibliography}
\end{document}